\documentclass[twoside,11pt]{article}
\usepackage{pgm}
\usepackage{amsmath}

\usepackage{hyperref,color,soul,booktabs}
\setulcolor{blue}
\newcommand{\BibTeX}{\textsc{B\kern-0.1emi\kern-0.017emb}\kern-0.15em\TeX}

\ShortHeadings{Parameterized Completeness Results for Bayesian Inference}{Bodlaender, Donselaar, Kwisthout}

\newcommand\restr[2]{{% we make the whole thing an ordinary symbol
  \left.\kern-\nulldelimiterspace % automatically resize the bar with \right
  #1 % the function
  \vphantom{\big|} % pretend it's a little taller at normal size
  \right|_{#2} % this is the delimiter
  }}

\begin{document}

% Title and authors
\title{Parameterized Completeness Results for Bayesian Inference}

\author{Hans L.~Bodlaender \Email{h.l.bodlaender@uu.nl}\\
   \addr Utrecht University, Department of Information and Computing Sciences\\
   Princetonplein 5, 3508 TB Utrecht, the Netherlands\and
   \Name{Nils Donselaar} \Email{nils.donselaar@donders.ru.nl}\\ 
   \Name{Johan Kwisthout}  \Email{johan.kwisthout@donders.ru.nl}\\
   \addr Radboud University, Donders Institute for Brain, Cognition and Behaviour\\
   Thomas van Aquinostraat 4, 6525 GD Nijmegen, the Netherlands}

\maketitle

% Abstract and keywords
\begin{abstract}%   <- trailing '%' for backward compatibility of .sty file
We present completeness results for inference in Bayesian networks with respect to two different parameterizations, namely the number of variables and the topological vertex separation number. For this we introduce the parameterized complexity classes $\mathsf{W[1]PP}$ and $\mathsf{XLPP}$, which relate to $\mathsf{W[1]}$ and $\mathsf{XNLP}$ respectively as $\mathsf{PP}$ does to $\mathsf{NP}$. The second parameter is intended as a natural translation of the notion of pathwidth to the case of directed acyclic graphs, and as such it is a stronger parameter than the more commonly considered treewidth. Based on a recent conjecture, the completeness results for this parameter suggest that deterministic algorithms for inference require exponential space in terms of pathwidth and by extension treewidth. These results are intended to contribute towards a more precise understanding of the parameterized complexity of Bayesian inference and thus of its required computational resources in terms of both time and space.
\end{abstract}

\section{Introduction}

Close to 35 years ago, \cite{LauritzenS88} revealed the central role of treewidth in the complexity of Bayesian inference, by exhibiting what we would call in modern terminology a fixed-parameter tractable algorithm for inference parameterized by the treewidth and maximum cardinality of the variables. While various other exact algorithms for inference have been developed since then, treewidth has remained a determining factor in their complexity, and indeed \cite{Kwisthout10} has shown that this dependence may well be unavoidable. However, relatively little is currently known about the precise complexity of inference with respect to treewidth, nor with respect to similar graph parameters for that matter. In this paper, we take the first steps in remedying this situation by providing such characterisations with respect to two different parameters, namely the number of variables (not to be confused with the aforementioned cardinality of individual variables) and the topological vertex separation number (serving as a proxy for pathwidth).\\
\\
The main contribution of this paper lies in identifying the parameterized complexity classes for which the corresponding parameterized inference problems are complete. Most of these complexity classes have not yet been explicitly considered in the literature, or only recently as in the case of the class $\mathsf{XNLP}$. The completeness results for the number of variables as a parameter mostly serve as a more accessible demonstration of the proof techniques which we employ for the topological vertex separation number. That said, we believe these results also have independent value as they offer a valuable perspective on the subsequent results, the number of variables being in some sense the strongest possible graph parameter.\\
\\
The completeness results for the topological vertex separation number are certainly of more direct importance. We introduce this parameter in Section~\ref{sec:tvsn} as a straightforward translation of the vertex separation number to the directed setting, the latter being equivalent to the pathwidth of a graph as shown in \cite{Kinnersley92}. In particular, we show completeness for (the probabilistic version of) $\mathsf{XNLP}$, a class which
was proposed by \cite{ElberfeldST15} and has been brought to attention in \cite{BodlaenderGNS22}. These results gain further relevance by relating them to a conjecture put forward in \cite{PilipczukW18}, which by the discussion in \cite{BodlaenderGNS22} can be rephrased as expressing that $\mathsf{XNLP}$-hard problems cannot be solved under simultaneous restrictions of time and space. More precisely, they conjecture that no such problem can have a (deterministic) algorithm which runs in both $|x|^{g(k)}$ time and $f(k)|x|^c$ space for some computable functions $f,g$ and constant $c$. In other words, if this conjecture holds, no $\mathsf{XNLP}$-hard problem has an $\mathsf{XP}$-algorithm which uses only a parameterized polynomial amount of space.\\
\\
Our results therefore imply that, absent any restriction on the cardinality of variables, inference parameterized by the number of variables lies at what is generally believed to be the lower bound of parameterized intractability, and (under the aforementioned conjecture) inference parameterized by the topological vertex separation number cannot be solved efficiently using only limited space. Though we do not explicitly demonstrate this here, the idea is that the latter parameter corresponds to the pathwidth of the network, which means that this result also holds for pathwidth and by extension for treewidth as well.

\section{Definitions}

In this section we cover the relevant basics from parameterized complexity theory, which includes definitions of particular kinds of reductions, complexity classes and problems complete for these. The paper's novel contributions  in this area are found in Sections~\ref{sec:numvar} and~\ref{sec:tvsn}.  

\begin{definition}\label{def:pareduc}
A \emph{parameterized reduction} from a parameterized problem $k_1$-$D_1$ to a parameterized problem $k_2$-$D_2$ is a mapping $(x,k) \mapsto (x',k')$ such that $(x,k)$ is a Yes-instance of $k_1$-$D_1$ precisely when $(x',k')$ is a Yes-instance of $k_2$-$D_2$, and $k_2 \leq g(k_1)$ for some computable function $g$. A \emph{fixed-parameter tractable reduction} (fpt-reduction) is a parameterized reduction which is computable in time $f(k)|x|^c$ for some computable function $f$ and constant $c$. A \emph{parameterized logspace reduction} (pl-reduction) is a parameterized reduction which is computable using space $f(k) + O(\log |x|)$ for some computable function $f$ and constant $c$.
\end{definition}

The complexity class $\mathsf{W[1]}$ is typically defined as the class of problems fpt-reducible to \textsc{Weighted $n$-Satisfiability} for fixed $n\geq 2$. Here we instead use the machine characterisation of $\mathsf{W[1]}$ which was established in  \cite{chenflumgrohe2}. A more natural complete problem for $\mathsf{W[1]}$ is $k$-\textsc{Clique} -- see e.g.~\cite{downeyfellows2} for a proof of this fact.

\begin{definition}\label{def:W1}
$\mathsf{W[1]}$ is the class of parameterized decision problems $k$-$D$ for which there exists a computable function $f$, a constant $c$ and a non-deterministic Turing machine $\mathcal{M}$ which on input $(x,k)$ correctly decides in time $f(k)|x|^c$ with $f(k)\log |x|$ tail-restricted non-determininism, i.e.~only the last $f(k)\log |x|$ steps are allowed to be non-deterministic.
\end{definition}
\noindent
$k$-\textsc{Clique}\\
\textbf{Input:} A graph $G=(V,E)$, an integer $k$.\\
\textbf{Parameter:} $k$.\\
\textbf{Question:} Is there a $W\subseteq V$ such that $|W|\geq k$ and $(u,v)\in E$ for all distinct $u,v\in W$?

\begin{lemma}\label{lem:kClique}
$k$-\textsc{Clique} is $\mathsf{W[1]}$-complete.
\end{lemma}

The complexity class $\mathsf{XNLP}$ was first introduced under this name in \cite{BodlaenderGNS22}, where it was shown amongst other things that $k$-\textsc{Chained Multicoloured Clique} is $\mathsf{XNLP}$-complete. We will also build on previous work on this class which has been carried out in \cite{ElberfeldST15}, which used the descriptive name $N[f \text{poly},f \log]$ instead. 

\begin{definition}\label{def:XNLP}
$\mathsf{XNLP}$ is the class of parameterized decision problems $k$-$D$ for which there exists a computable function $f$, a constant $c$ and a non-deterministic Turing machine $\mathcal{M}$ which on input $(x,k)$ correctly decides in time $f(k)|x|^c$, using at most $f(k)\log |x|$ space.
\end{definition}
\noindent
$k$-\textsc{Chained Multicoloured Clique}\\
\textbf{Input:} A graph $G=(V,E)$, where $V$ is partitioned into sets $V_1,\ldots,V_r$ such that for every $(u,v)\in E$, if $u\in V_i$ and $v\in V_j$, then $|i-j|\leq 1$, and a colouring function $f: V\rightarrow [1,k]$.\\
\textbf{Parameter:} $k$.\\
\textbf{Question:} Is there a $W\subseteq V$ such that for each $i\in [1,r-1]$, $W \cap (V_i\cup V_{i+1})$ is a clique, and for every $i\in[1,r]$ and $j\in[1,k]$ there is a $w\in W\cap V_i$ with $f(w) = j$? 

\begin{lemma}\label{lem:kCMC}
$k$-\textsc{Chained Multicoloured Clique} is $\mathsf{XNLP}$-complete.
\end{lemma}

\section{Parameterization by the Number of Variables}\label{sec:numvar}

While the number of variables of the network is an unsuitable parameter in practice, the results in this section are a useful reminder of the importance of the size of the probability distributions, which is what prevents tractability with respect to this parameter. In what follows we will first consider the special case of \textsc{Positive Inference} (i.e.~whether the probability is non-zero) before moving on to the general problem of \textsc{Bayesian Inference}.

\subsection{\textsf{W[1]}-completeness of \textit{n}-\textsc{Positive Inference}}

Here we show that $n$-\textsc{Positive Inference}, formally described below, is $\mathsf{W[1]}$-complete.\\
\\
$n$-\textsc{Positive Inference}\\
\textbf{Input:} A Bayesian network $\mathcal{B} = (G,\text{Pr})$, two sets of variables $\mathbf{H},\mathbf{E}\subseteq V$ along with joint value assignments $\mathbf{h}\in\Omega(\mathbf{H})$ and $\mathbf{e}\in\Omega(\mathbf{E})$.\\
\textbf{Parameter:} The number of variables $n = |V|$.\\
\textbf{Question:} Does $\text{Pr}(\mathbf{h}\mid\mathbf{e})>0$ hold?

\begin{theorem}\label{thm:nPInf}
$n$-\textsc{Positive Inference} is $\mathsf{W[1]}$-complete.
\end{theorem}
\begin{proof}Our proof proceeds by giving fpt-reductions from and to the $\mathsf{W[1]}$-complete problem of $k$-\textsc{Clique} to show $\mathsf{W[1]}$-hardness and membership respectively.
\\
\\
For the reduction from $k$-\textsc{Clique} we create $k$ variables $X_1,\ldots,X_k$ each of which is uniformly distributed over the $n$ vertices of the graph $G$. For every pair $X_i,X_j$ we create an additional binary variable $X_{i,j}$ as its common child, such that $\text{Pr}(X_{i,j}= \text{True} \mid X_i=u,\ X_j=v) = 1$ if $(u,v)\in E$ and $\text{Pr}(X_{i,j}= \text{True} \mid X_i=u,\ X_j=v) = 0$ otherwise. This adds $\binom{k}{2}\leq k^2$ additional variables, each with conditional probability tables with $n^2$ entries.\\
\\
Finally, we create a binary variable $X_C$ with all variables $X_{i,j}$ as its parents, such that $X_C$ is True with probability 1 whenever all $X_{i,j}$ are True, and False with probability 1 otherwise. This means that $X_C$ will have a conditional probability table with at most $2^{k^2}$ entries. Thus we have created a Bayesian network with $O(k^2)$ variables and total description size $O(2^{k^2} + k^2n^2)$ with the property that $\text{Pr}(X_C=\text{True})>0$ if and only if $G$ has a $k$-clique.
\\
\\
For the reduction to $k$-\textsc{Clique} we create for every variable $X$ a vertex for each tuple of assignments to $\pi_+(X)$ which is assigned non-zero probability in its conditional probability table. Let $S(X)$ be the set of these vertices. If $X,Y$ are two different variables, add an edge between $u\in S(X)$ and $v\in S(Y)$ unless there is some variable $Z\in \pi_+(X)\cap\pi_+(Y)$ to which the corresponding tuples assign different values. Now $\text{Pr}(x_1,\ldots,x_n)>0$ implies that $\bigcup_{i=1}^n\{\restr{(x_1,\ldots,x_n)}{\pi_+(X_i)}\}$ is an $n$-clique in the graph constructed in this way.\footnote{Note that the number of variables $n$ is now our parameter $k$ for the resulting graph which serves as an instance of $k$-\textsc{Clique}, whose number of vertices will be in the order of the number of entries across all conditional probability tables in the Bayesian network.}\\
\\
Conversely, an $n$-clique can only arise by choosing one vertex from each set $S(X_i)$ such that the corresponding tuples are consistent and hence give rise to a joint value assignment $(x_1,\ldots,x_n)$ with its probability the product of $n$ non-zero probabilities. To restrict to the question whether $\text{Pr}(\mathbf{H}=\mathbf{h}\mid \mathbf{E}=\mathbf{e}) > 0$, we furthermore remove all vertices from the graph for which the corresponding tuple assigns different values to $\mathbf{H}$ and $\mathbf{E}$ than those specified. 
\end{proof}

We now move on to the general problem of \textsc{Bayesian Inference}.

\subsection{\textsf{W[1]PP}-completeness of \textit{n}-\textsc{Bayesian Inference}}

We begin by introducing the $\mathsf{PP}$-equivalent of $\mathsf{W[1]}$, which we call $\mathsf{W[1]PP}$ following the naming convention proposed in \cite{ElberfeldST15}, or alternatively $\mathsf{PFPT[1]}$ according to the notation of \cite{montoyamuller}.

\begin{definition}\label{def:W1PP}
$\mathsf{W[1]PP}$ is the class of parameterized decision problems $k$-$D$ for which there exists a computable function $f$, a constant $c$ and a probabilistic Turing machine $\mathcal{M}$ which on input $(x,k)$ halts in time $f(k)|x|^c$, using at most $f(k)\log |x|$ tail-restricted non-deterministic steps, with probability strictly more than $\frac{1}{2}$ of giving the correct answer.
\end{definition}

As is the case for $\mathsf{PP}$, this definition is equivalent to the one where we also allow the probability of rejecting a No-instance to be exactly $\frac{1}{2}$. As we can use this fact to streamline some of the upcoming proofs, we will formally establish this as the following lemma.

\begin{lemma}\label{lem:eqdef1}
Definition~\ref{def:W1PP} yields the same complexity class when $\mathcal{M}$ accepts Yes-instances with probability strictly more than $\frac{1}{2}$ and No-instances with at probability at most $\frac{1}{2}$.
\end{lemma}

\begin{proof}
Suppose we have such a machine $\mathcal{M}$. Since $\mathcal{M}$ takes at most $f(k)\log|x|$ non-deterministic steps, we know that its probability of accepting Yes-instances will be at least $\frac{1}{2} + |x|^{-f(k)}$. We can modify $\mathcal{M}$ to obtain a machine $\mathcal{M}'$ which works like $\mathcal{M}$ up until it would halt. At this point, $\mathcal{M}'$ preserves rejections, but where $\mathcal{M}$ would accept it takes another $f(k)\log|x| + 1$ non-deterministic steps in order to achieve a probability of $\frac{1}{2}|x|^{-f(k)}$ of changing the acceptance into a rejection. As a result, $\mathcal{M}'$ accepts Yes-instances with probability at least $\frac{1}{2}+\frac{1}{4}|x|^{f(k)}$ and No-instances with probability at most $\frac{1}{2}-\frac{1}{4}|x|^{f(k)}$. Moreover, this modification does not violate the limitations of the tail-restricted bounded non-determinism, which shows the two definitions to be equivalent.
\end{proof}

Another canonical $\mathsf{W[1]}$-complete problem is \textsc{Short Turing Machine Acceptance} -- we again refer to \cite{downeyfellows2} for a proof. We show next that the analogous problem of \textsc{Short Turing Machine Majority Acceptance} is $\mathsf{W[1]PP}$-complete, so that we can subsequently use it in a reduction in order to demonstrate $\mathsf{W[1]PP}$-hardness.\\
\\
\textsc{Short Turing Machine Majority Acceptance (STMMA)}\\
\textbf{Input:} A non-deterministic Turing machine $\mathcal{M}$, a string $x$, an integer $k$.\\
\textbf{Parameter:} $k$.\\
\textbf{Question:} Does $\mathcal{M}$ accept on $x$ within $k$ steps with probability strictly greater than $\frac{1}{2}$?

\begin{proposition}\label{prop:STMMA}
\textsc{STMMA} is $\mathsf{W[1]PP}$-complete.
\end{proposition}

\begin{proof} 
For membership we simulate the first $k$ steps of $\mathcal{M}$ on $x$ using $O(k \log n)$ steps of non-determinism, accepting in case $\mathcal{M}$ accepted $x$ within that time and rejecting otherwise.\\
\\
To establish hardness, we consider the majority acceptance problem for an arbitrary $\mathsf{W[1]}$-machine $\mathcal{M}$ (i.e~a Turing machine which terminates in $f(k)|x|^c$ steps with $f(k)\log|x|$ tail-restricted non-deterministic steps) on input $x$. The parameterized reduction from this problem to \textsc{STMMA} can be characterised in terms of \textit{simulation} followed by \textit{compression}.\\
\\
Since $\mathcal{M}$'s non-determinism is tail-restricted, it is possible to deterministically simulate its execution on $|x|$ as part of the parameterized reduction up until the point where the computation enters its non-deterministic tail of $f(k)\log |x|$ steps. To eliminate this factor $\log|x|$, the second part of the reduction compresses $\mathcal{M}$'s alphabet and transition functions, so that the resulting machine needs only $g(k)$ steps to complete the computation.\\
\\
Thus the parameterized reduction transforms $\mathcal{M},x$ into an instance $\mathcal{M}',x'$ of \textsc{STMMA} where $\mathcal{M}'$ is a compressed version of $\mathcal{M}$ initialised to match $\mathcal{M}'s$ state upon entering the non-deterministic tail, and similarly $x'$ corresponds to the compressed tape contents at that point, with the parameter value being the number of non-deterministic steps $g(k)$. Given that the simulation and compression are fixed-parameter tractable and the correctness is directly ensured through its construction, this yields a valid fpt-reduction. 
\end{proof}

At this point one may wonder why we did not consider the majority variant of a satisfiability problem, as is perhaps more commonly done. Indeed, given that \textsc{Weighted 2-Satisfiability} is canonically $\mathsf{W[1]}$-complete, surely we could have consider \textsc{Weighted Majority 2-Satisfiability} instead? The issue however is that this problem appears not to be $\mathsf{W[1]PP}$-hard.\footnote{A similar failure of a majority satisfiability problem to be hard for what would initially appear to be the corresponding parameterized complexity class was observed in \cite{donselaar2019}.} We believe that the recent results on \textsc{Majority $k$-Satisfiability} in \cite{akmalwilliams} and subsequent work point to an underlying reason for this failure. This may possibly be addressed in future work -- for now, we proceed with establishing the $\mathsf{W[1]PP}$-completeness of $n$-\textsc{Bayesian Inference}, which refers to the following problem.\\
\\
$n$-\textsc{Bayesian Inference}\\
\textbf{Input:} A Bayesian network $\mathcal{B} = (G,\text{Pr})$, two sets of variables $\mathbf{H},\mathbf{E}\subseteq V$ along with joint value assignments $\mathbf{h}\in\Omega(\mathbf{H})$ and $\mathbf{e}\in\Omega(\mathbf{E})$ and a probability threshold $q\in\mathbb{Q}$. We require that all probabilities (including $q$) are represented as a unary fraction, and furthermore that the probabilities of any specific distribution have the same denominator.\\
\textbf{Parameter:} The number of variables $n = |V|$.\\
\textbf{Question:} Does $\text{Pr}(\mathbf{h}\mid\mathbf{e})>q$ hold?
\begin{theorem}\label{thm:nBInf}
$n$-\textsc{Bayesian Inference} is $\mathsf{W[1]PP}$-complete.
\end{theorem}

\begin{proof}
For membership, we provide a straightforward and commonly-used argument which generates a single instantiation of the network by means of forward sampling.\\
\\
We sample each variable according to a topological ordering by randomly generating an assignment according to the local probability distribution (conditioned on the relevant previously generated assignments where necessary). That is, if we wish to sample from a probability distribution $\text{Pr}(X)$ such that for every $x_i$ its probability is stored as a unary fraction with denominator $D$, we non-deterministically generate an integer $s\in[1,D]$ and take the assignment $x_k$ where $k$ is the largest value such that $\Sigma_{i<k}\text{Pr}(x_i) < \frac{s}{D} \leq \Sigma_{i\leq k}\text{Pr}(x_i)$. Because of the unary representation of $D$, generating $s$ can be done using logarithmic space as required. Once this procedure is completed, we have a full instantiation of the network.\\
\\
In the case without evidence we accept with probability $1-\frac{q}{2}$ if the sample agrees with $\mathbf{h}$, and with probability $\frac{1}{2}-\frac{q}{2}$ if it does not, which requires additional non-deterministic steps based on the size of $q$. This process leads to a probability of acceptance of $\frac{1}{2} + \frac{1}{2}(\text{Pr}(\textbf{h})-q)$, hence we accept with probability more than $\frac{1}{2}$ precisely when $\text{Pr}(\textbf{h}) > q$ as required.\\
\\
For the case with evidence, we need to distinguish between $q\geq\frac{1}{2}$ and $q < \frac{1}{2}$. If $q\geq\frac{1}{2}$, we accept with probability $\frac{1}{2q}$ if the sample agrees with both $\mathbf{h}$ and $\mathbf{e}$ and with probability $\frac{1}{2}$ whenever the sample does not agree with $\mathbf{e}$, rejecting otherwise. This leads to a probability of accepting equal to $\frac{1}{2}+\frac{\text{Pr}(\mathbf{e})}{2q}(\text{Pr}(\mathbf{h}\mid\mathbf{e})-q)$.
If instead $q < \frac{1}{2}$, we accept with probability 1 if the sample agrees with both $\mathbf{h}$ and $\mathbf{e}$, with probability $\frac{1-2q}{2-2q}$ if the sample agrees with $\mathbf{e}$ but not with $\mathbf{h}$, and with probability $\frac{1}{2}$ whenever the sample does not agree with $\mathbf{e}$. This leads to a probability of accepting equal to $\frac{1}{2}+\frac{\text{Pr}(\mathbf{e})}{2-2q}(\text{Pr}(\mathbf{h}\mid\mathbf{e})-q)$. Under both circumstances we again accept with probability more than $\frac{1}{2}$ precisely when $\text{Pr}(\textbf{h}\mid\mathbf{e}) > q$ as required.\\
\\
To establish hardness we shall provide a Cook-style construction similar to the one used in \cite{donselaar2019}, reducing from the problem \textsc{STMMA}. Since we essentially want to simulate the machine $\mathcal{M}$ for at most $k$ steps, as a consequence we will use at most $k$ space. Thus we can explicitly represent the full computation by tracking the $k$ tape cells across all $k+1$ time steps for a total of $O(k^2)$ variables. We describe the network in more detail below.\\
\\
First of all, we have for every $i\in [0,k]$ and $j\in [1,k]$ a node $S_{i,j}$ which tracks for the $i$-th time step whether the tape head is currently at the $j$-th tape cell, what the contents of this cell are, and which state the machine is locally ``believed'' to be in. For all $j\in [1,k]$ and $i\neq k$, the node $S_{i,j}$ has a single child $T_{i,j}$ which stores the new contents and state, and furthermore an instruction for moving the head out of $\{-,\text{left},\text{stay},\text{right}\}$. Each node $T_{i,j}$ has up to three nodes as its children, namely $S_{i+1,j-1}$, $S_{i+1,j}$ and $S_{i+1,j+1}$. Finally, the nodes $S_{k,j}$ have a single child $D_j$ which can be either True or False; any $D_j$ is also a parent of $D_{j+1}$.\\
\\
The nodes $S_{0,j}$ are initialised with the correct information, i.e.~contents, head location and state. If $S_{i,j}$ indicates that the head is not currently at that cell, the distribution for $T_{i,j}$ is such that it will deterministically copy the symbol $\sigma$ and state $A$ and add the instruction $-$. However, if $S_{i,j}$ indicates that the head is currently at that cell, $T_{i,j}$ will randomly select one of the possible transitions by taking on the resulting symbol $\sigma'$, state $A'$ and instruction from $\{\text{left},\text{stay},\text{right}\}$ (or simply `stay' in the accepting state).\\
\\
For any $i\neq 0$, the distribution for $S_{i,j}$ is defined first of all to copy over the symbol from $T_{i-1,j}$. Furthermore, if $T_{i-1,j-1}$ contains the instruction `right', $T_{i-1,j}$ contains `stay' or $T_{i-1,j+1}$ contains `left', $S_{i,j}$ will note the head is now at its location and copy over the state information from the corresponding $T$-variable. In all other cases, $S_{i,j}$ will note the head is not currently at its location and copy the state from $T_{i-1,j}$.\\
\\
The preceding construction only guarantees that the state information is correct at nodes where the head is located at that time. Thus $D_{j+1}$ is defined to be True precisely when $S_{k,j}$ has the head and is in the accepting state \emph{or} $D_{j-1}$ is True (if it exists), and False otherwise. This ensures that $\text{Pr}(D_k = \text{True})$ is the probability of $\mathcal{M}$ accepting on $x$ within $k$ steps, hence to complete the reduction we need to set $q=\frac{1}{2}$. To see that this is an fpt-reduction, note that the graph has $2k(k+1)$ variables with indegree at most 3 and which take a value out of at most twice the number of states times the size of the tape alphabet many options.     
\end{proof}

\section{Parameterization by the Topological Vertex Separation Number}\label{sec:tvsn}

For an intuition behind the topological vertex separation number, suppose we want to efficiently use working memory while generating a sample of a Bayesian network. One trick we can use is to only keep the sampled value of a particular variable in working memory until we assigned values to all of its dependent variables, at which point we can write it off to storage memory. This is essentially how a logspace transducer operates: it has a read-only input tape, a read/write work tape of logarithmic size, and a write-only output tape. Now the topological vertex separation number of a directed acyclic graph expresses the upper bound on the number of variables for which we need to have its value stored at any point while generating the sample. As we shall demonstrate in this section, we can use it to derive a parameterized logspace bound on the space required to perform inference.

\begin{definition}\label{def:t}
Let $G = (V,A)$ be a directed acyclic graph. Given a topological ordering $T$ of $V$, define $V_T(i) = \{u\in V\mid T(u) < i \wedge \exists v((u,v)\in A \wedge T(v)\geq i)\}$. Let $t_T(G) = \text{max}_i|V_T(i)|$. The \emph{topological vertex separation number} $t(G)$ of $G$ is given by $\text{min}_Tt_T(G)$.
\end{definition}

This definition translates that of the vertex separation number for an undirected graph to directed acyclic graphs, with the added requirement that the permutation is a topological ordering. Note that the usual definition of vertex separation requires $T(u) \leq i$ and $T(v) > i$ instead: writing $W_T(i)$ for the corresponding set, we see $W_T(i) = V_T(i+1)$. As a consequence, the topological vertex separation number of a directed acyclic graph is an upper bound on the vertex separation number of its moralization.\footnote{Suppose the moralization added an edge $(T(i),T(j))$; let $i<j$ without loss of generality, so that $T(i)$ occurs in $W_T(s)$ for all $s\in [i,j-1]$. This edge is only added if there is some $l > j > i$ such that $(T(i),T(l)),(T(j),T(l))\in A$, hence $T(i)$ already occurred in these sets before the moralization.} By \cite{Kinnersley92}, it is known that the vertex separation number of an undirected graph is equal to its pathwidth, and in turn pathwidth is an upper bound on treewidth. Thus all of the hardness results in this section also apply with relation to these graph parameters -- we aim to identify the corresponding complexity classes for the parameterization by treewidth in future work.

\subsection{\textsf{XNLP}-completeness of \textit{t}-\textsc{Positive Inference}}

As in Section~\ref{sec:numvar}, we first study \textsc{Positive Inference} for this parameterization before proceeding with inference in its full generality. That is, we consider the following problem:\\
\\
$t$-\textsc{Positive Inference}\\
\textbf{Input:} A Bayesian network $\mathcal{B} = (G,\text{Pr})$, two sets of variables $\mathbf{H},\mathbf{E}\subseteq V$ along with joint value assignments $\mathbf{h}\in\Omega(\mathbf{H})$ and $\mathbf{e}\in\Omega(\mathbf{E})$.\\
\textbf{Parameter:} The topological vertex separation number $t(G)$.\\
\textbf{Question:} Does $\text{Pr}(\mathbf{h}\mid\mathbf{e})>0$ hold?

\begin{theorem}\label{thm:tPInf}
$t$-\textsc{Positive Inference} is $\mathsf{XNLP}$-complete.
\end{theorem}
\begin{proof} Similar to the proof of Theorem~\ref{thm:nPInf}, we proceed by giving pl-reductions from and to the $\mathsf{XNLP}$-complete problem of $k$-\textsc{Chained Multicoloured Clique}.
\\
\\
For the reduction from $k$-\textsc{Chained Multicoloured Clique}, we create a network as follows. Write $S_{i,j} = \{v\in V_i \mid f(v) = j\}$ and $S_i = \bigcup_{j=1}^k S_{i,j}$. First, we create variables $X_{i,j}$ which are uniformly distributed over $S_{i,j}$. Let $(i,j) < (i',j')$ if $i<i'$ or $i=i'$ and $j<j'$. For every pair of variables $X_{i,j},X_{i',j'}$ from $S_i \cup S_{i+1}$ such that $(i,j) < (i',j')$ we now create a variable $C_{i,j,i',j'}$ such that $\text{Pr}(C_{i,j,i',j'} = \text{True}\mid X_{i,j} = u, X_{i',j'} = v) = 1$ if $(u,v)\in E$ and $\text{Pr}(C_{i,j,i',j'} = \text{True}\mid X_{i,j} = u, X_{i',j'} = v) = 0$ otherwise. We order the variables such that $C_{i,j,i',j'} < C_{a,b,a',b'}$ if $(i',j') < (a',b')$, or $(i',j') = (a',b')$ and $(i,j) < (a,b)$.\\
\\
We then create variables $D_{i,j,i',j'}$, which we order similarly, as ANDs of $C_{i,j,i',j'}$ and the preceding $D$-variable ($D_{1,1,1,2}$ has a single parent $C_{1,1,1,2}$). This construction ensures that $\text{Pr}(D_{r,k-1,r,k} = \text{True}) > 0$ precisely when the original graph contains a chained multicoloured clique. Moreover, consider the following topological ordering $T$:
\begin{itemize}
\item $X_{i,j} <_T X_{i',j'}$ if $(i,j)<(i',j')$.
\item $C_{i,j,i',j'} <_T C_{a,b,a',b'}$ whenever $(i',j') < (a',b')$, or $(i',j') = (a',b')$ and $(i,j) < (a,b)$.
\item $D_{i,j,i',j'} <_T D_{a,b,a',b'}$ whenever $(i',j') < (a',b')$, or $(i',j') = (a',b')$ and $(i,j) < (a,b)$.
\item $X_{i,j} <_T C_{a,b,a',b'}$ (same for $D_{a,b,a',b'}$) if $(i,j) \leq (a',b')$.
\item $C_{i,j,i',j'} <_T X_{a,b}$ (same for $D_{i,j,i',j'}$) whenever $X_{a,b} \not<_T C_{i,j,i',j'}$.
\item $C_{i,j,i',j'} <_T D_{a,b,a',b'}$ whenever $C_{i,j,i',j'} <_T C_{a,b,a',b'}$ or $(i,j,i',j') = (a,b,a',b')$.
\item $D_{i,j,i',j'} <_T C_{a,b,a',b'}$ whenever $C_{a,b,a',b'} \not<_T D_{i,j,i',j'}$.
\end{itemize}
We see that $V_T(T(X_{i,j}))$ consists of those $X_{i',j'}$ such that $(i-1,1)\leq (i',j') < (i,j)$ plus the most recent $D$-variable. Similarly, $V_T(T(C_{i,j,i',j'}))$ and $V_T(T(D_{i,j,i',j'}))$ consist of at most of those $X_{a,b}$ such that $(i'-1,1) \leq (a,b) \leq (i',j')$ plus the preceding $D$-variable. Thus the topological vertex separation number of the underlying graph is $2k$. By constructing the instance according to the given topological ordering, each individual variable can be constructed using at most $O(k\log n)$ space (though obviously it cannot be \emph{represented} using logarithmic space), which means this is indeed a pl-reduction.\\
\\
For the reduction to $k$-\textsc{Chained Multicoloured Clique}, given an instance of $t$-\textsc{Positive Inference}, let $T$ be a topological ordering such that $t_T(G) = t$. For every $i\in[1,n]$, let $V^*_T(i)$ be obtained from $V_T(i)$ by adding the number of copies of the variable $X$ such that $T(X) = i$ which is needed to ensure that $|V^*_T(i)| = t+1$ (to guarantee that at least one copy is added). We number the variables from 1 to $t+1$ according to their order in $T$.\\
\\
For every variable $X$ from $V^*_T(i)$ we create a vertex for each tuple of assignments to $\pi_+(X)$ which is assigned non-zero probability in its conditional probability table. These vertices receive the colour $j$ which is the number just given to $X$ in $V^*_T(i)$. Let $V_i$ be the set of these vertices. For every pair of differently coloured vertices $u,v\in V_i$, add $(u,v)$ to $E$ unless there is some variable to which $u$ and $v$ assign different values. We do the same for every pair of vertices $u\in V_i$ and $v\in V_{i+1}$, regardless of their colours.\\
\\
As a consequence, $\text{Pr}(x_1,\ldots,x_n)>0$ implies that the set which arises from taking for every $X$ the vertices corresponding to the tuple $\restr{(x_1,\ldots,x_n)}{\pi_+(X_i)}$ is a chained multicoloured clique in the graph $(V_1,\ldots,V_n,E)$ with $t+1$ colours. Conversely, any chained multicoloured clique in this graph can only arise in this way. To restrict to the question whether $\text{Pr}(\mathbf{H}=\mathbf{h}\mid \mathbf{E}=\mathbf{e}) > 0$, we once again omit from the graph all vertices for which the corresponding tuple assigns different values to $\mathbf{H}$ and $\mathbf{E}$ than those specified. This is again a pl-reduction as $t$ limits the number of variables on the basis of which we construct an individual vertex or edge of the resulting graph.
\end{proof}

\subsection{\textsf{XLPP}-completeness of \textit{t}-\textsc{Bayesian Inference}}

Our approach here is much like in the previous section: we introduce the class $\mathsf{XLPP}$, show that it is well-behaved, and define a machine acceptance problem which can serve as an intermediate step in the proof that $t$-\textsc{Bayesian Inference} is hard for this class.

\begin{definition}\label{def:XLPP}
$\mathsf{XLPP}$ is the class of parameterized decision problems $k$-$D$ for which there exists a computable function $f$, a constant $c$ and a probabilistic Turing machine $\mathcal{M}$ which on input $(x,k)$ halts in time $f(k)|x|^c$, using at most $f(k)\log |x|$ space, with probability strictly more than $\frac{1}{2}$ of giving the correct answer.
\end{definition}

Once again, this definition is equivalent to the one where we also allow the probability of rejecting a No-instance to be exactly $\frac{1}{2}$.

\begin{lemma}\label{lem:eqdef2}
Definition~\ref{def:XLPP} yields the same complexity class when $\mathcal{M}$ accepts Yes-instances with probability strictly more than $\frac{1}{2}$ and No-instances with at probability at most $\frac{1}{2}$.
\end{lemma}

\begin{proof}
The argument is essentially the same as for the proof of Lemma~\ref{lem:eqdef1}: though the minimum bound away from $\frac{1}{2}$ is smaller, we can simply take $f(k)|x|^c+1$ non-deterministic steps, as we have more than enough space to keep track of how many we took.   
\end{proof}

As with \textsc{STMMA}, we arrive at a suitable machine acceptance problem by taking the majority version of a problem shown to be complete for $\mathsf{XNLP}$ in \cite{ElberfeldST15}.\\
\\
\textsc{Timed Space-bounded Turing Machine Majority Acceptance (TSTMMA)}\\
\textbf{Input:} A non-deterministic Turing machine $\mathcal{M}$, unary integers $s,t$.\\
\textbf{Parameter:} $s$.\\
\textbf{Question:} Does $\mathcal{M}$ accept with probability strictly greater than $\frac{1}{2}$ within $t$ steps on an initially empty tape using at most $s$ tape cells?

\begin{proposition}\label{prop:TSTMMA}
\textsc{TSTMMA} is $\mathsf{XLPP}$-complete.
\end{proposition}

\begin{proof} 
We refer to Theorem 4.4 in \cite{ElberfeldST15} for a proof that the problem \textsc{Timed Space-bounded Turing Machine Acceptance} (there called $p_s$-\textsc{Timed NSC}) is $\mathsf{XNLP}$-complete. One can easily observe that said proof requires no essential adjustments in order to be suitable for our current purposes (cf.~the proof of Proposition~\ref{prop:STMMA}).  
\end{proof}

All that is left at this point is to present the definition of $t$-\textsc{Bayesian Inference} and the proof that it is $\mathsf{XLPP}$-complete. For membership we again require that the probabilities are represented as unary fractions -- we will return to this point in the conclusions.\\
\\
$t$-\textsc{Bayesian Inference}\\
\textbf{Input:} A Bayesian network $\mathcal{B} = (G,\text{Pr})$, two sets of variables $\mathbf{H},\mathbf{E}\subseteq V$ along with joint value assignments $\mathbf{h}\in\Omega(\mathbf{H})$ and $\mathbf{e}\in\Omega(\mathbf{E})$ and a probability threshold $q\in\mathbb{Q}$. We require that all probabilities (including $q$) are represented as a unary fraction, and furthermore that the probabilities of any specific distribution have the same denominator.\\
\textbf{Parameter:} The topological vertex separation number $t(G)$.\\
\textbf{Question:} Does $\text{Pr}(\mathbf{h}\mid\mathbf{e})>q$ hold?

\begin{theorem}\label{thm:tBInf}
$t$-\textsc{Bayesian Inference} is $\mathsf{XLPP}$-complete.
\end{theorem}

\begin{proof}
The argument for membership proceeds along the same lines as in the proof of Theorem~\ref{thm:nPInf}, although now we need to manage our space more carefully. While sampling according to a topological ordering $T$ such that $t_T(G) = t$, we check after every new value whether the sample thus far still agrees with $\mathbf{h}$ and $\mathbf{e}$ respectively. This allows us to forget values of variables once their children have been assigned values as well, which ensures that we need at most $t\log|x|$ space at any point during the construction of the sample.\\
\\
To establish $\mathsf{XLPP}$-hardness, we reduce from \textsc{TSTMMA} again using mostly the same construction as in the proof of Theorem~\ref{thm:nPInf}. The only difference (apart from the variables $S_{0,j}$ being initialised as empty) is that we have variables for $i\in[0,t]$ and $j\in [1,s]$, hence there are $2s(t+1)$ variables in total. By choosing the straightforward topological ordering which orders the variables within each row in ascending order before proceeding to the next row, we see that $t(G) \leq s+1$ as required. To see that this is indeed a pl-reduction, note that in particular we can compute the correct probabilities using only logarithmic space by employing a more efficient (binary) representation before writing them out in unary.
\end{proof}

\section{Conclusions}

The $\mathsf{XNLP}$-completeness of $t$-\textsc{Positive Inference} (Theorem~\ref{thm:tPInf}) can be considered the primary result of this paper for a number of reasons. From a negative point of view, the $\mathsf{XLPP}$-completeness of $t$-\textsc{Bayesian Inference} (Theorem~\ref{thm:tBInf}) may be considered somewhat artificial due to the technical requirement that the probabilities are represented in unary, which was necessary to ensure parameterized logspace computability. However, from a positive point of view, the $\mathsf{XNLP}$-completeness of $t$-\textsc{Positive Inference} tells us far more than may be obvious at first glance. Combined with the observation that the topological vertex separation number of a graph is an upper bound on the pathwidth and in particular the treewidth of the moralized graph, and the conjecture from \cite{PilipczukW18} regarding the deterministic space complexity of $\mathsf{XNLP}$-hard problems, this result suggests a superpolynomial lower bound (with respect to any of these graph parameters) on the space required to perform even a limited form of inference. Our hope is that this insight contributes to an increased understanding of the (parameterized) complexity of Bayesian inference, and that it serves as a starting point for further investigations along these lines.

\bibliography{PCRfBI}

\end{document}